\documentclass{article}
\usepackage{fullpage}
\usepackage{authblk}
\usepackage{algorithm,algorithmicx,algpseudocode}
\usepackage{amsthm,amsfonts,amsmath}
\usepackage{xspace}
\usepackage[utf8]{inputenc}
\usepackage{verbatim}
\usepackage{enumerate}
\usepackage[capitalize]{cleveref}
\newcommand{\dotdot}{\ldots}

\theoremstyle{theorem}
\newtheorem{theorem}{Theorem}[section]
\newtheorem{proposition}[theorem]{Proposition}

\newtheorem{lemma}[theorem]{Lemma}
\newtheorem{problem}{Problem}
\newtheorem{observation}[theorem]{Observation}

\theoremstyle{remark}
\newtheorem{example}[theorem]{Example}

\newcommand{\Oh}{\mathcal{O}}
\newcommand{\Par}{\mathcal{P}}
\newcommand{\Cells}{\mathsf{Cells}}
\newcommand{\cell}{\mathsf{cell}}
\newcommand{\Blocks}{\mathsf{Blocks}}
\newcommand{\integ}{\mathbb{Z}}
\newcommand{\floor}[1]{\lfloor #1 \rfloor}
\newcommand{\PP}{\mathbf{P}}
\newcommand{\sub}{\subseteq}

\newcommand{\defproblem}[3]{
\begin{problem}{#1}

\noindent
\textbf{Input:} #2

\noindent
\textbf{Output:} #3
\end{problem}
}

\newcommand{\R}{\mathcal{R}}
\newcommand{\I}{L}
\renewcommand{\P}{\mathcal{P}}

\title{On Abelian Longest Common Factor\\ with and without RLE}

 \author[1]{Szymon Grabowski}
 \author[2]{Tomasz Kociumaka}
 \author[2]{Jakub Radoszewski}

 \affil[1]{Łódź University of Technology, Łódź, Poland, e-mail:~\texttt{sgrabow@kis.p.lodz.pl}}
 
 \affil[2]{University of Warsaw, Warsaw, Poland, e-mail:~\texttt{kociumaka@mimuw.edu.pl}, \texttt{jrad@mimuw.edu.pl}}

\begin{document}
\maketitle
\begin{abstract}
\noindent
We consider the Abelian longest common factor problem in two scenarios: when input strings are uncompressed and are of size $n$, and when the input strings are run-length encoded and their compressed representations have size at most $m$. The alphabet size is denoted by $\sigma$.
For the uncompressed problem, we show an $o(n^2)$-time and $\Oh(n)$-space algorithm in the case of $\sigma=\Oh(1)$, making a non-trivial use of tabulation.
For the RLE-compressed problem, we show two algorithms: one working in $\Oh(m^2\sigma^2 \log^3 m)$ time and $\Oh(m (\sigma^2+\log^2 m))$ space, which employs line sweep, and one that works in $\Oh(m^3)$ time and $\Oh(m)$ space that applies in a careful way a sliding-window-based approach. The latter improves upon the previously known $\Oh(nm^2)$-time and $\Oh(m^4)$-time algorithms that were recently developed by Sugimoto et al.\ (IWOCA 2017) and Grabowski (SPIRE 2017), respectively.

\end{abstract}

\section{Introduction}
\noindent
Two strings are called \emph{Abelian equivalent} if one of them is a permutation of the other. A string $p$ is called an \emph{Abelian factor} of a string $u$ if it is Abelian equivalent to one of the factors of $u$. Our aim in this work is to compute the longest common Abelian factor of two strings, $u$ and $v$. The longest common Abelian factor is an approximate similarity measure of strings in the scope of so-called non-standard stringology.

The longest common Abelian factor problem is denoted here as LCAF. We also consider a version of this problem, denoted as RLE-LCAF, in which the strings are specified by their run-length encodings (called here RLE representations). We denote: by $n$ the length of the strings, by $m$ the length of their RLE representations, and by $\sigma$ the size of the alphabet.

\subsection{Previous Results}
\paragraph{Related Abelian stringology problems}
The best studied problem in Abelian stringology is the \emph{jumbled indexing problem}. In this problem we are to index a text to support queries asking if a given string is an Abelian factor of the text. The query string is represented as a \emph{Parikh vector} which stores the number of occurrences of each letter from the alphabet in the pattern. In the case of a binary alphabet $\Sigma=\{0,1\}$, Cicalese et al.~\cite{PSC2009-10} proposed an index with $\Oh(n)$ size and $\Oh(1)$ query time and gave an $\Oh(n^2)$-time construction algorithm for the index. The key observation behind their index is that it suffices to store, for every query length $\ell$, the minimum and maximum number of ones in a factor of length $\ell$ of the text.

After a series of works of Burcsi et al.\ \cite{DBLP:conf/fun/BurcsiCFL10,DBLP:journals/ijfcs/BurcsiCFL12} and Moosa and Rahman~\cite{DBLP:journals/ipl/MoosaR10,DBLP:journals/jda/MoosaR12}, the construction of a binary jumbled index was improved to $\Oh(\frac{n^2}{(\log{n})^2})$. Furthermore, Hermelin et al.\ \cite{DBLP:journals/corr/HermelinLRW14} reduced binary jumbled indexing to all-pairs shortest paths problem and obtained preprocessing time of $\Oh(\frac{n^2}{2^{\Omega((\log n/\log \log n)^{0.5})}})$ (a similar reduction was shown by Bremner at el.\ \cite{DBLP:journals/algorithmica/BremnerCDEHILPT14}). Finally, Chan and Lewenstein~\cite{DBLP:conf/stoc/ChanL15} used techniques from additive combinatorics to improve the construction time of the binary index to $\Oh(n^{1.859})$. Subquadratic-time and space constructions of a jumbled index for any constant-sized alphabet were proposed in \cite{DBLP:journals/algorithmica/KociumakaRR17,DBLP:conf/stoc/ChanL15}.

Binary jumbled indexing was also considered in the case that the text is given as its RLE representation of length $m$. Constructions of the index working in $\Oh(n+m^2 \log m)$ time~\cite{DBLP:journals/tcs/AmirAHLLR16,DBLP:journals/ipl/BadkobehFKL13} and in $\Oh(n+m^2)$ time~\cite{DBLP:conf/cpm/CunhaDGWKS17,DBLP:journals/ipl/GiaquintaG13} were proposed.

As for other Abelian stringology problems, subquadratic-time algorithms for computing Abelian squares, Abelian periods, Abelian runs, Abelian covers, and Abelian borders over a constant-sized alphabet were designed in~\cite{DBLP:conf/macis/KociumakaRW15,AIMS}. Computation of Abelian borders, Abelian periods, and Abelian squares on strings specified by their RLE representations was considered in \cite{DBLP:journals/tcs/AmirAHLLR16,DBLP:journals/corr/SugimotoNIBT17}.

\paragraph{Longest common Abelian factor}
In the special case of a binary alphabet, the LCAF problem reduces in linear time to binary jumbled indexing \cite{DBLP:journals/ijfcs/AlatabbiILR16}. Indeed, it suffices to construct jumbled indexes of each of the strings and then to check, for each length $\ell$, if both strings contain an Abelian factor of length $\ell$ containing the same number of ones. Thus binary LCAF can be solved in $\Oh(n^{1.859})$ time using using the best known jumbled index \cite{DBLP:conf/stoc/ChanL15}. Moreover, binary RLE-LCAF can be solved in $\Oh(n+m^2)$ time and $\Oh(n)$ space by applying an efficient binary jumbled index for an RLE representation of the text~\cite{DBLP:conf/cpm/CunhaDGWKS17,DBLP:journals/ipl/GiaquintaG13}.

Over a general alphabet, for the LCAF problem the fastest known algorithms work in $\Oh(n^2 \sigma)$ time and $O(n)$ space, and in $\Oh(n^2 \log^2 n \log^{*} n)$ time and $O(n \log^2 n)$ space~\cite{DBLP:conf/spire/BadkobehGGNPS16}.

Known solutions for the RLE-LCAF problem (for arbitrary $\sigma$) work in $\Oh(nm^2)$ \cite{DBLP:journals/corr/SugimotoNIBT17}, 
in $\Oh(m^4)$, 
and in $\Oh(n^{3/2}\sigma \sqrt{m\log n})$ (provided that $m = \Oh(n / \log n)$) time \cite{DBLP:conf/spire/Grabowski17}, respectively.

\subsection{Our results}
\noindent
We first consider the LCAF problem when $\sigma$ is $\Oh(1)$.
Although subquadratic-time (and sometimes even $\Oh(n^{2-\varepsilon})$-time) algorithms are known for many Abelian stringology problems in the case of a constant-sized alphabet, no such algorithm was previously developed for the longest common Abelian factor problem. Moreover, the reduction to the jumbled indexing problem does not work for alphabet size $\sigma>2$. We present the first $o(n^2)$-time algorithms for LCAF with any $\sigma=\Oh(1)$. We first describe algorithms that work in $\Oh(n^2 / \log^{1/\sigma} n)$ time and $\Oh(n (\log\log n)^2 /\log n)$ time, and then combine both techniques to obtain $\Oh(n^2 / \log^{1+1/\sigma} n)$ time complexity. Our algorithms work in $\Oh(n)$ space. This approach is described in Section~\ref{sec:1}.

In Section~\ref{sec:red} we show a reduction of RLE-LCAF to a problem of intersecting rectangles in $\mathbb{Z}^\sigma$. This allows us to develop two solutions to RLA-LCAF, that work in:
\begin{itemize}
    \item $\Oh(m^2\sigma^2 \log^3 m)$ time and $\Oh(m (\sigma^2+\log^2 m))$ space (see Section~\ref{sec:2}), and
    \item $\Oh(m^3)$ time and $\Oh(m)$ space (see Section~\ref{sec:3}).
  \end{itemize}
  The latter improves upon the time complexities of the algorithms of Sugimoto et al.~\cite{DBLP:journals/corr/SugimotoNIBT17} ($\Oh(nm^2)$) and Grabowski~\cite{DBLP:conf/spire/Grabowski17} ($\Oh(m^4)$).
In the case of constant $\sigma$, we obtain the following improved versions of the former:
  \begin{itemize}
    \item in $\Oh(m^2 \sqrt{\log \log m})$ time in expectation
  or $\Oh(m^2 \log \log m)$ time deterministically and $\Oh(m)$ space for $\sigma=2$, and
    \item in $\Oh(m^2 \log^2 m)$ time and $\Oh(m \log m)$ space for $\sigma=3$.
  \end{itemize}

\section{Preliminaries}
\noindent
We assume that the symbols of a string are numbered starting from $1$.
A factor of string $u$ spanning from position $i$ to position $j$ (inclusive) will be denoted as $u[i \ldots j]$.
The string $u$ is a concatenation of symbols over an alphabet $\Sigma = \{1, 2, \ldots, \sigma\}$.
The concatenation of two strings, $u$ and $v$, is denoted as $u v$.
By $\P(u)$ we denote the Parikh vector of a string $u$.
It is defined as a vector (array) of size $\sigma$ storing the number of occurrences of each alphabet symbol in $u$.
Formally, $\P(u)[c] = k$ iff $|\{i: u[i] = c\}| = k$, for any alphabet symbol $c$.
Two Parikh vectors are equal when the equality of corresponding counters holds for all symbols from $\Sigma$.
We also define $\PP(u)$ to be the family of Parikh vectors of all factors of $u$.
Recall that our task is to find a vector $P\in \PP(s)\cap \PP(t)$ maximizing $\|P\|_{\ell_1}$.

The run-length encoding (RLE) representation of string $u$ of length $n$ is a sequence of $m$ non-empty substrings $u_i$, $1 \leq i \leq m$, such that $u = u_1 u_2 \ldots u_m$, the number of distinct alphabet symbols in each $u_i$ is one, and the number of distinct symbols in every concatenation $u_i u_{i+1}$ is two.
It is trivial to obtain the RLE representation of $u$ in $\Oh(n)$ time, but in the (RLE-related) algorithms presented in this work we assume that the input strings are already RLE-compressed.
The RLE representation can be stored in $\Oh(m)$ space.
In this work, the RLE representations of strings are denoted by capital letters.

The notation $u \sim v$ tells that the strings $u$ and $v$ are Abelian equivalent.
We say that string $p$ is an Abelian factor of string $u$ if there exist indices $i$ and $j$ such that $u[i \ldots j] \sim p$.
A common Abelian factor of two strings, $u$ and $v$, 
is a pair of factors $u[i' \ldots j']$ and $v[i'' \ldots j'']$ 
such that $u[i' \ldots j'] \sim v[i'' \ldots j'']$ 
(obviously, $j' - i' = j'' - i''$).


All logarithms considered in this work are of base 2.

Let us formally state the problems studied in this work.

\defproblem{LCAF}{
  two strings $s$ and $t$ over an alphabet of size $\sigma = \Oh(1)$, each of length at most $n$

}{
  the length of the longest common Abelian factor of $s$ and $t$

}

\defproblem{RLE-LCAF}{
  RLE representations $S$ and $T$ of two strings $s$ and $t$ over an alphabet of size $\sigma = \Oh(1)$, each representation of length at most $m$ and each string of length at most $n$

}{
  the length of the longest common Abelian factor of $s$ and $t$

}

We assume the word-RAM model with machine words of $w = \Theta(\log n)$ bits.

\newcommand{\tsort}{T_{\mathrm{sort}}}
\newcommand{\ssort}{S_{\mathrm{sort}}}

  Let $\tsort(m)$ and $\ssort(m)$ denote the time and space to sort $m$ integers (in the word-RAM). Currently the fastest randomized algorithm works in $\Oh(m^2 \sqrt{\log \log m})$ time in expectation \cite{DBLP:conf/focs/HanT02} and the fastest deterministic algorithm works in $\Oh(m^2 \log \log m)$ worst-case time \cite{DBLP:journals/jal/Han04}. Both algorithms require $\Oh(m)$ space.

\section{Algorithm for LCAF over Constant-Sized Alphabet}\label{sec:1}
\noindent
In this section, we present three slightly subquadratic algorithms for LCAF with constant-size alphabets.
In this setting, \cite[Sec.~4]{DBLP:conf/spire/BadkobehGGNPS16} gives an algorithm with $O(n^2)$ time and $O(n)$ space.

The input consists in two strings, $s$ and $t$, of total length $n$, with symbols over an alphabet 
$\Sigma$ of size $\sigma=\Oh(1)$.
We assume that the symbols of $s$ and $t$ are packed into $\Oh(n \log \sigma/w)$ machine words
so that any $\Oh(w/\log\sigma)$ consecutive characters can be retrieved in constant time.

We start by introducing the subdivision of the space $\integ^\sigma$ into cells, which is the main concept common to our algorithms.
Next, in \cref{sec:one}, we present a simple solution running in $\Oh(n^2 / \log^{1/\sigma} n)$ time,
which is improved to $\Oh((n\log\log n)^2 /\log n)$ in \cref{sec:two}.
Finally, in \cref{sec:three}, we derive an $\Oh(n^2 / \log^{1+1/\sigma} n)$ bound on the running time.
For super-constant alphabet size $\sigma$, the time complexities of the first two algorithms
increase by factors polynomial in $\sigma$ (which we analyze in detail),
while for the last solution the extra factor is exponential in $\sigma$ (and we omit the detailed analysis).

\subsection{Orthogonal Cells in $\integ^{\sigma}$}\label{sec:common}
Note that $\PP(x)\sub \integ^{\sigma}$; the key tool in our algorithms is a subdivision of $\integ^\sigma$ into orthogonal \emph{cells}. 
For a vector $P=(p_1,\ldots,p_{\sigma})$ and a positive integer $b$,
we define 
$$\floor{P/b} = (\floor{p_1/b},\ldots,\floor{p_\sigma/b})\quad\quad\text{and}\quad\quad P \bmod b = (p_1 \bmod b,\ldots,p_\sigma \bmod b).$$
We define cells (of \emph{side length} $b$) as equivalence classes with respect to the mapping $P\leadsto \floor{P/b}$.
The family of all such cells is denoted by $\Cells_b$, and $\cell_b : \integ^{\sigma}\to \Cells_b$ is the canonical projection.
In the algorithms, a cell $C\in \Cells_b$ is identified by the value $\floor{P/b}$ common to all $P\in C$,
and any vector $P\in C$ is identified by $P\bmod b$.

\subsection{$\Oh(n^2 / \log^{1/\sigma} n)$ time}\label{sec:one}
\noindent
Our solution uses a parameter $b\ge \sigma$, whose value will be settled later. 
We process the strings $s$ and $t$ in $\Oh(n/b)$ \emph{stages};
each stage is responsible for factors of length within a range $R$ of size $|R|\le b$. In other words,
our task is to find the maximum common Abelian factor of $s$ and $t$ whose length belongs to 
$R$ or to certify that there is no such common Abelian factor.

The main mechanism used by our algorithm is a simple bucketing: 
for each considered factor $u$, its Parikh vector $\P(u)$ will be 
inserted into a bucket corresponding to the cell $\cell_b(\P(u))$.
Then, we shall scan all non-empty buckets in search of a vector inserted both as an Abelian factor of $s$ and of $t$.

In the first solution, we store the contents of each bucket simply as a bitmask of size $b^\sigma$ (equal to the cell size).
We require that $b^\sigma \leq w = \Theta(\log n)$, which implies $b = \Oh(\log^{1/\sigma} n)$.
As a result of processing $s$, for each cell $C$ we shall guarantee that in the corresponding bucket the bit representing $Q\in \{0,\ldots,b-1\}^\sigma$
is set if and only if $s$ contains a factor $u$ with $|u|\in R$, $\cell_b(\P(u))=C$, and $\P(u)\bmod b = Q$.
The other string $t$ is handled in the same way; for clarity, below we discuss processing $s$ only.

First, we scan $s$ in order to construct a list of $\Oh(n\sigma)$ \emph{requests} to insert certain vectors to certain buckets.
A single request consists of bucket's identifier and a bitmask representing Parikh vectors to be inserted there.
There might be many requests concerning the same bucket, but we shall make sure that after all these insertions
are performed, the contents of each bucket are as specified above.

In the $j$-th step, we consider all factors $u$ (with $|u|\in R$) starting at position $j$,
and our aim is to create insertion requests responsible for their Parikh vectors.
Let these factors be $u^{(0)},\ldots,u^{(b')}$ (ordered by increasing lengths) for $0\le b'<b$,
and let $B^{(i)}=\floor{\P(u^{(i)})/b}$.
Note that for each coordinate $d$, we have $$B^{(0)}[d] \le \cdots \le B^{(b'-1)}[d] \le B^{(0)}[d]+1.$$
Consequently, the sequence $B^{(0)},\ldots,B^{(b')}$ consists of at most $\sigma+1$ distinct cells.

Furthermore, we note that the shift $B^{(i)}-B^{(0)}$ and the vector $\P(u^{(i)})\bmod b$
depend only on the vector $\P(u^{(0)})\bmod b$ and the last characters of $u^{(1)},\ldots,u^{(b')}$.
Therefore, we can build a lookup table whose keys consist of
\begin{enumerate}[(a)]
  \item a vector $Q\in \{0,\ldots,b-1\}^\sigma$, corresponding to $\P(u^{(0)})\bmod b$, and
  \item up to $b-1$ symbols $c_1,\ldots,c_{b'}$, corresponding to the last characters of $u^{(1)},\ldots,u^{(b')}$.
\end{enumerate}
As the values, we store up to $\sigma+1$ insertion requests to buckets,
with cell identifiers stored relative to $B^{(0)}$.
The key size is thus $\Oh(\sigma \cdot \log b + (b-1) \cdot \log\sigma) = \Oh(b\log \sigma)=o(\log n)$ bits, 
while the value contains $\Oh(\sigma)$ machine words.
Hence, the lookup table can be constructed in $o(n\sigma)$ time (and it can be used across all stages).

In the $j$-th step, we retrieve the necessary insertion requests from the lookup table
and we add $B^{(0)}$ to shift the cell identifiers.
As a result, the pass over $s$ produces $\Oh(n\sigma)$ insertion requests to buckets,
representing the Parikh vectors of all substrings $u$ of $s$ with $|u|\in R$.
The string $t$ is processed analogously.

Recall that our task is to decide if the requests from $s$ and $t$ contain a common entry.
To verify this, we group the requests by the cell identifiers and process each cell independently.
A single cell identifier takes $\Oh(\sigma \log n)$ bits and there are $\Oh(n\sigma)$ requests to be grouped,
so this process can be implemented in $\Oh(n\sigma^2)$ time using radix sort.

For each cell, we build the corresponding buckets, separately for $s$ and $t$.
A single request is handled in constant time with a simple bitwise-OR operation on two machine words.

After that, to check if the two buckets contain a common entry,
we perform a bitwise-AND operation on the two bitmasks.
When the result of this operation is non-zero, we find in constant time 
(e.g., using another lookup table) a set bit representing a Parikh vector with maximum $\ell_1$ norm
(as our goal is obviously to find the longest common Abelian factor).

The total running time of the presented algorithm is $\Oh(n^2\sigma^2 / b) = \Oh((n\sigma)^2 / \log^{1/\sigma} n)$ 
and the space consumption is $\Oh(n\sigma^2)$ words.

\subsection{$\Oh(n (\log\log n)^2 /\log n)$ time}\label{sec:two}
\noindent
Recall that the cell size is $b^\sigma$. Hence,
its elements can be represented using $\log (b^\sigma)=\sigma \log b$ bits each.
In this solution, we change the bucket representation to a packed list (see~\cite[Fact~5.1]{DBLP:journals/algorithmica/KociumakaRR17}),
which is simply a concatenation of the $(\sigma \log b)$-bit integers representing its contents (possibly with repetitions).

We also use this representation in the insertion requests stored in the lookup table.
The key size is still $\Oh(b\log \sigma)$ bits,
while the value size is now $\Oh(\sigma^2 + b\cdot \sigma \log b)=\Oh(\sigma b\log b)$ bits.
We take $b= o(\log n/ \log \sigma)$ to make sure that the table size and construction time are $o(n\sigma)$.

The total size of all $\Oh(n\sigma)$ requests constructed in a single stage is now $\Oh(n  b \sigma \log b)$ bits.
As each bucket is represented using a packed list, concatenation is used to create
a packed list representing it (with entries coming from one or more requests).

To answer LCAF, for each cell (with non-empty buckets) we need to check if the buckets constructed for $s$ and $t$ contain a common entry. 
To this end, we use Lemma~5.3 from~\cite{DBLP:journals/algorithmica/KociumakaRR17},
which lets us compute for a given packed list the $FirstOcc$ bitmask, 
which marks positions where each entry occurs for the first time in the packed list.
We construct $FirstOcc$ bitmasks $F_1$ and $F_2$ for the packed lists representing the two buckets, and a bitmask $F_3$ for the concatenation of those two lists.
Finally, we observe that the buckets have no element in common if and only if $F_3 = F_1 F_2$.

Let us now analyze the time and space complexity of the described variant.
In a single stage, we have $\Oh(n\sigma)$ packed lists with $\Oh(nb)$ entries
in total, and the universe size is $N=b^\sigma$.
By \cite[Lemma~5.3]{DBLP:journals/algorithmica/KociumakaRR17},
the $FirstOcc$ bitmasks can be computed in $\Oh(n\sigma+ nb\log^2(b^\sigma)/w)=\Oh(n\sigma + nb\sigma^2\log^2 b/w)$ time,
while the space complexity is $\Oh(n\sigma + n\sigma b\log b/w)$ words.
Across all stages, the overall running time becomes $\Oh(n^2\sigma^2 /b + n^2\sigma^2 \log^2 b / w)$,
whereas the space consumption is $\Oh(n\sigma^2 + n\sigma b\log b/w)$ words.
Setting $b=\Theta(\log n/ \log \log n)$, we obtain the promised $\Oh( (n\sigma \log\log n)^2 /\log n)$ time
using $\Oh(n\sigma^2)$ words of space. 

\subsection{$\Oh(n^2 / \log^{1+1/\sigma} n)$ time}\label{sec:three}
In our final solution, instead of using a single partition of $\integ^\sigma$ into cells,
we recursively subdivide $\integ^\sigma$ into cells of smaller and smaller side length.
For each cell $C$, we solve the LCAF problem restricted to $C$, i.e.,
we find a vector  $P\in C\cap \PP(s)\cap \PP(t)$ maximizing $\|P\|_{\ell_1}$.
Depending on the side length $b$ and the size of the corresponding buckets,
we either solve this task directly, or we partition $C$ into $2^\sigma$ smaller cells
and recurse on each of them.

Our main improvement compared to \cref{sec:two} is a more space-efficient encoding of $\PP_C(v):= C\cap \PP(v)$ for fixed $C$.
To develop it, we also recursively subdivide $\mathbb{N}$ into \emph{blocks}:
for a parameter $b$, the blocks $\Blocks_b$ are consecutive intervals of length $b$ (the last block might be shorter).
For a string $v$ and two blocks $I,J\in \Blocks_b$, we define 
$$\PP_{(I,J)}(v)= \{\Par(v[i\dotdot j]) : i\in I\text{ and }j\in J\}.$$
For a cell $C\in \Cells_b$, we also define $\Blocks_C(v)$ as the set of all pairs $(I,J)\in \Blocks_b^2$
such that $C$ intersects the \emph{bounding box} of $\PP_{(I,J)}(v)$.

Let us fix $C\in \Cells_b$ and a string $v$.
For any $([i\dotdot i'],[j\dotdot j'])=(I,J)\in \Blocks_C(v)$, we keep $C\cap\PP_{(I,J)}(v)$ in $\Oh(b)$ bits as follows.
We store $\Par(v[i'\dotdot j])$ relative to $C$ (which takes $\Oh(\sigma\log b)=\Oh(\log b)$ bits)
as well as the characters $v[i],\ldots,v[i']$ and $v[j],\ldots,v[j']$ (which take $\Oh(b\log \sigma)=\Oh(b)$ bits).
The set $\PP_C(v)$ is then simply kept as a concatenation of the representations of $C\cap \PP_{(I,J)}(v)$
over $(I,J)\in \Blocks_C(v)$.
The size of this representation is $\Oh(1+|\Blocks_C(v)|b/\log n)$ machine words.
Moreover, it is easy to construct it in $\Oh(n)$ time for $b=n$ (and the block $[1,\ldots,n]$).

To solve the problem for a cell $C\in \Cells_b$, we consider three cases.
If $b^\sigma < w=\Theta(\log n)$, we convert the representations of $\PP_C(s)$ and $\PP_C(t)$
into bitmasks: we scan them word by word, use a lookup table to convert each word into a bitmask,
and combine these bitmasks with bitwise-OR. Finally, we bitwise-AND the bitmasks obtained for $s$ and $t$.

On the other hand, if $\PP_C(s)$ and $\PP_C(t)$ take $\log n$ bits in total,
we use another precomputed table to extract the answer.

In the remaining cases, we partition $C$ into $2^\sigma$ cells $C'\in \Cells_{b/2}$.
For each such cell $C'$, we scan the representation of $\PP_C(v)$
and construct an analogous representation of $\PP_{C'}(v)$:
for each $(I,J)\in \Blocks_C(v)$, we consider all four pairs of blocks $I',J'\in \Blocks_{b/2}$
with $I'\sub I$ and $J'\sub J$, construct  the representation of $C'\cap \PP_{(I',J')}(v)$,
and append it to the representation of $\PP_{C'}(v)$ provided that $(I',J')\in \Blocks_{C'}(v)$.
If $b=\Omega(\log n)$, it is easy to process each pair $(I,J)\in \Blocks_{C}(v)$ in $\Oh(b/\log n)$ time.
For $b=o(\log n)$, on the other hand, we build a lookup table to exploit bit parallelism.

We conclude with a complexity analysis.
Observe that $(I,J)\in \Blocks_{C}(v)$ for at most $3^\sigma = \Oh(1)$ cells $C\in \Cells_b$,
so the total size of the representations of $\PP_C(v)$ in a single level (for fixed $b$) is $\Oh(n^2/b)$ bits.
Since we terminate the recursion whenever $\PP_C(s)$ and $\PP_C(t)$ contain $\log n$ bits,
the processing time is $\Oh(n^2/(b\log n))$ per level.
This bound forms a geometric progression dominated by the largest term
$\Oh(n^2/\log^{1+1/\sigma} n)$ arising from $b=\Theta(\log^{1/\sigma} n)$.

The space complexity is at most $\Oh(n)$ bits within each recursive call,
because $|\Blocks_C(v)|=\Oh(n/b)$ for $C\in \Cells_b$.
Overall, this gives $\Oh(n\log n)$ bits, i.e., $\Oh(n)$ machine words.

\section{RLE-LCAF as a Problem of Intersecting Rectangles}\label{sec:red}

\newcommand{\rect}{\mathit{rect}}
\newcommand{\Rect}{\mathit{Rect}}
\newcommand{\placeholder}{\ensuremath{\star}\xspace}

\noindent
In this section we show a reduction of RLE-LCAF to a problem
of intersecting rectangles in the $\sigma$-dimensional space $\integ^{\sigma}$. This reduction is then used in both the algorithms for RLE-LCAF in the next two sections. We also develop basic properties of the resulting rectangle sets.

We define a \emph{rectangle in $d$-dimensional space $\integ_+^d$} ($d \ge 2$) as a Cartesian product of $d$ closed intervals, such that at least $d-2$ of them are singletons.
E.g., $\{3\} \times [2,5] \times [1,7] \times \{0\}$ is a rectangle in $\integ_+^4$.

For an RLE-representation $V$ of string $v$ and indices $i,j$ such that $1 \le i \le j \le |V|$, we denote by $\rect_V(i,j)$ a rectangle with opposite corners $\P(V_i \ldots V_j)$ and $\P(V_{i+1} \ldots V_{j-1})$. If $i=j$ or $i+1=j$, the latter is the zero vector. Let
  $$\Rect_V=\{\rect_V(i,j)\,:\,1 \le i \le j \le |V|\}.$$
  
  \begin{observation}\label{obs:rect_simple}
    The integer points in rectangles from $\Rect_V$ represent the set $\PP(v)$.
  \end{observation}
    
      This observation lets us reduce the RLE-LCAF problem to the following auxiliary problem.

\defproblem{Maximal Intersection Point of Rectangles in $\integ_+^d$\label{prob:rect}}{
  two families $\R_1$ and $\R_2$ of rectangles in $d$-dimensional space, each containing at most $N$ rectangles

}{
  a common point of a rectangle from $\R_1$ and a rectangle from $\R_2$ with the maximum $\ell_1$-norm
  or ``NO'' if no two rectangles from $\R_1$ and $\R_2$ intersect

}

\begin{lemma}\label{lem:red}
  The RLE-LCAF problem is equivalent to Problem~\ref{prob:rect} with $N=m^2$, $d=\sigma$, $\R_1=\Rect_S$, and $\R_2=\Rect_T$.
\end{lemma}
\begin{proof}
By Observation~\ref{obs:rect_simple}, the points in rectangles from $\Rect_S$ and $\Rect_T$ represent Parikh vectors of all factors of $s$ and $t$, respectively.
  Hence, the point returned by the solution to Problem~\ref{prob:rect} for $\R_1$ and $\R_2$ represents the Parikh vector of the longest common Abelian factor of $s$ and $t$.
\end{proof}

\subsection{Properties of Rectangles in $\Rect_V$}

  For a rectangle $R$, by $\I(R)$ we denote the interval of $\ell_1$-norms of points in $R$. For an integer $l$, by $\Rect_V(l)\sub \Rect_V$ we denote a subset which consists of rectangles $R$ such that $l \in \I(R)$.
  
      For a given index $i \in \{1,\ldots,m\}$, by $j(i,l)$ we denote the minimum index $j$ such that $|V_i|+ \ldots+ |V_j| \ge l$. If no such index exists, we set $j(i,l)=m+1$. We further set $j(m+1,l)=m+1$. Indices $j(i,l)$ allow us to characterize the set $\Rect_V(l)$ as follows.
    
    \begin{observation}\label{obs:rect-l}
      $\Rect_V(l) = \{\rect_V(i,j)\,:\,i=1,\ldots,m,\,j(i,l) \le j \le j(i+1,l),\,j \le m\}$.
    \end{observation}
  
  The following lemma states some algorithmic properties of the sets $\Rect_V(l)$.
  
  \begin{lemma}\label{lem:rect-l}
    Let $V$ be an RLE representation of size $m$ of a string $v$.
    \begin{enumerate}[(a)]
      \item\label{aaa} For a given $l \in \{0,\ldots,|v|\}$, the set $\Rect_V(l)$ has at most $2m$ elements and can be computed in $\Oh(m)$ time.
      \item\label{bbb} All sets $\Rect_V(l) \setminus \Rect_V(l-1)$ and $\Rect_V(l-1) \setminus \Rect_V(l)$ for $l=1,\ldots,|v|$ such that at least one of these sets is non-empty can be computed in $\Oh(m\tsort(m))$ total time and $\Oh(\ssort(m))$ space.
    \end{enumerate}
  \end{lemma}
  \begin{proof}
    \eqref{aaa}: We use the characterization of $\Rect_V(l)$ from the observation. Indices $j(i,l)$ can be computed in $\Oh(m)$ time using a sliding-window-based approach. Note that
    $$|\Rect_V(l)| \le \sum_{i=1}^m (j(i+1,l)-j(i,l)+1) = m + \sum_{i=1}^m (j(i+1,l)-j(i,l)) \le 2m.$$
    This also implies that $\Rect_V(l)$ can be computed in $\Oh(m)$ time.
    
    \eqref{bbb}: We store the current set $\Rect_V(l)$ in a data structure $S$ which is an array indexed by $i$ of lists of $\rect_V(i,j)$, ordered by $j$ in each list. Each rectangle is represented in $\Oh(1)$ space.
    
    For subsequent values of $l$ we store all pairs of the form $(|V_i|+\ldots+|V_{j(i,l)}|,i)$ such that $j(i,l)<m+1$ in a priority queue, with the minimum stored on the top. Let $(a,i)$ be the pair currently stored on the top. If $l<a$, we know that $\Rect_V(l)=\ldots=\Rect_V(a)$ but $\Rect_V(a+1) \ne \Rect_V(a)$. Thus we will increase $l$ directly to $l=a+1$. To compute the symmetric difference of the two sets, we pop from the priority queue all pairs with the first component equal to $a$. For each such pair $(a,i)$, we set $j(i,a+1)=j(i,a)+1$, removing $\rect_V(i+1,j(i,a))$ from $S$ and inserting $\rect_V(i,j(i,a+1))$ to $S$ if $j(i,a+1) \le m$. We then insert the pair $(|V_i|+\ldots+|V_{j(i,a+1)}|,i)$ to the priority queue provided that $j(i,a+1) \le m$.
   
   Each index $j(i,l)$ is incremented at most $m$ times. Hence:
   \begin{itemize}
     \item the total number of operations performed on the priority queue,
     \item the total number of lengths $l$ such that $(\Rect_V(l) \setminus \Rect_V(l-1)) \cup (\Rect_V(l-1) \setminus \Rect_V(l))$ is nonempty, and
     \item the total number of rectangles reported
   \end{itemize}
   are all bounded by $\Oh(m^2)$. By a reduction of Thorup~\cite{DBLP:journals/jacm/Thorup07}, the priority queue can be implemented in $\Oh(\ssort(m))$ space using $\Oh(\tsort(m)/m)$ time per operation.
  \end{proof}
  
  \section{Algorithm for RLE-LCAF over Small Alphabet}\label{sec:2}

\newcommand{\up}{\mathrm{up}}
\newcommand{\down}{\mathrm{down}}

\subsection{RLE-LCAF over Binary Alphabet}
First, we present a simple solution for $\sigma=2$. Our approach is based on the following known property specific to binary strings (see~\cite{PSC2009-10}).
\begin{observation}\label{obs:binary}
If  $(p,q_1),(p,q_2)\in \PP(v)$ for $q_1\le q_2$ and a binary string $v$, 
then $(p,q')\in \PP(v)$ for every integer $q'\in [q_1\dotdot q_2]$.
\end{observation}

In other words, the set $\PP(v)$ is an orthogonally convex subset of $\integ^2$.
Let us define $\up_v(p)=\max \{q : (p,q)\in \PP(v)\}$ and $\down_v(p)=\min\{q : (p,q)\in \PP(v)\}$
to be functions representing the upper and lower boundary of this region, respectively.
Note that due to \cref{obs:binary} (with the two coordinates interchanged)
and the fact that each point in $\PP(v)$ is dominated by $\Par(v)$, both these functions are non-decreasing.

\begin{lemma}\label{lem:updown}
Let $V$ be the RLE representation of a binary string $v$.
If the size of $V$ is $m$,
then the functions $\up_v$ and $\down_v$ are piecewise constant with $\Oh(m^2)$ pieces.
Such representations can be generated in the left-to-right
order in $\Oh(m\cdot \tsort(m))$ time using $\Oh(\ssort(m))$ space.
\end{lemma}
\begin{proof}
The function $\up_v$ is the upper envelope of $\Oh(m^2)$ rectangles in $\Rect_V$.
To determine $\up_v$, we process the top-left corners of the rectangles
in the left-to-right order. During this sweep, the value $\up_v$ is the largest second coordinate
of the already scanned points.
In a similar way, we determine $\down_v$ as the bottom envelope of the rectangles,
i.e., processing their bottom-right corners in the bottom-to-top order.

A simple implementation involves sorting the coordinates of all $\Oh(m^2)$ vectors $\P(V_i\cdots V_j)$. However, this would require $\Oh(m \ssort(m))$ space. Therefore, instead we maintain a priority queue of size $m$ which contains, for each index $i$, the coordinates of the top-left corner of the rectangle $\rect_V(i,j)$ such that $j$ is the first unprocessed index for this value of $i$. Upon the removal of $\rect_V(i,j)$ from the queue, we insert $\rect_V(i,j+1)$ provided that $j < m$. Finally, just as in the proof of Lemma~\ref{lem:rect-l}\eqref{bbb}, we use the reduction of Thorup~\cite{DBLP:journals/jacm/Thorup07} to implement the priority queue in $\Oh(\ssort(m))$ space using $\Oh(\tsort(m)/m)$ time per operation.
\end{proof}

\begin{proposition}\label{prp:rlcaf2}
  The RLE-LCAF problem for $\sigma=2$ can be solved in $\Oh(m\tsort(m))$ time and $\Oh(\ssort(m))$ space.
  With the state-of-the-art sorting algorithms, the running time is $\Oh(m^2 \sqrt{\log \log m})$ in expectation
  or $\Oh(m^2 \log \log m)$ deterministic, both with $\Oh(m)$ space.
\end{proposition}
\begin{proof}
We apply \cref{lem:updown} to determine the staircase functions $\up_s$, $\down_s$, $\up_t$, and $\down_t$.
Next, we scan their representations to find out for each $p$ if there is a point $(p,q)\in \PP(s)\cap \PP(t)$,
i.e., whether $\up_s(p)\ge \down_t(p)$ and $\up_t(p)\ge \down_s(p)$. If so, to maximize
the $\ell_1$ norm, we take $q=\min(\up_s(p),\up_t(p))$. We also observe that
it suffices to consider values $p$ which are right endpoints of a step in at least one of the considered functions.
The running time of this post-processing is $\Oh(m^2)$ and the extra space consumption is constant.

The overall time complexity is dominated by sorting integers in \cref{lem:updown}.
The fastest randomized~\cite{DBLP:conf/focs/HanT02} and deterministic~\cite{DBLP:journals/jal/Han04} sorting algorithms yield the announced
bounds on the running time.
\end{proof}

\subsection{Space-Efficient Reduction to Problem~\ref{prob:rect}}
\noindent
We use a more sophisticated reduction than the one of Lemma~\ref{lem:red} that leads to more space-efficient algorithms.
\begin{lemma}\label{lem:red2}
In $\Oh(m \tsort(m))$ time and $\Oh(\ssort(m))$ space,
the RLE-LCAF problem can be reduced to $\Oh(m)$ instances of Problem~\ref{prob:rect} with $N=\Oh(m)$, $d=\sigma$,
$\R_1\sub \Rect_S$, and $\R_2\sub \Rect_T$.
Each rectangle is represented in constant space as $R=\rect_V(i,j)$ for $V\in \{S,T\}$.
\end{lemma}
\begin{proof}
For subsequent integers $l=0$ to $\min(|s|,|t|)$ we maintain
families $\Rect_S(l)\sub \Rect_S$ and $\Rect_T(l)\sub \Rect_T$ using Lemma~\ref{lem:rect-l}\eqref{bbb}. We only consider the values of $l$ for which $\Rect_S(l)$ or $\Rect_T(l)$ changes comparing to $l-1$.

At the same time, we maintain sets $\R_1$ and $\R_2$ with $\Rect_S(l)\sub \R_1 \sub \Rect_S$ and $\Rect_T(l)\sub \R_2 \sub \Rect_T$.
To make sure that the invariant is satisfied, every insertion to $\Rect_S(l)$ and $\Rect_T(l)$ is performed also in $R_1$ and $R_2$, respectively.
Moreover, after every $m$ insertions (including the final insertion), we make an instance of \cref{prob:rect} out of $\R_1$ and $\R_2$,
and we set $\R_1 := \Rect_S(l)$ and $\R_2 := \Rect_T(l)$.
This way, the size of these families is bounded by $\Oh(m)$. Moreover, the number of insertions is $\Oh(m^2)$, so the number of instances is $\Oh(m)$.

Finally, note that the conclusion follows from the fact that for every $l$, including the length of the sought LCAF,
 we output an instance with $\Rect_S(l)\sub \R_1$ and $\Rect_T(l)\sub \R_2$.
\end{proof}

We say that an instance of Problem~\ref{prob:rect} is \emph{normalized} if all coordinates of the corners of the rectangles are of magnitude $\Oh(N)$.
Any instance of Problem~\ref{prob:rect} can be transformed into a normalized one in $\Oh(\tsort(N))$ time by renumbering each coordinate of the rectangles' corners
separately preserving their relative order.
In all the algorithms below we normalize the instance as a first step.
However, in order to compute the result for the original instance, which need not be normalized, when comparing
the $\ell_1$ norms of intersection points found, we need to transform the renumbered values of all coordinates back to the original ones.

Let us denote the dimensions by $x_1,\ldots,x_d$.
We define the \emph{dimensions of a rectangle} $R$ as the indices $i$ such that the projection of $R$ to $x_i$ is not a singleton.
We also say that rectangle is a \emph{$ij$-rectangle} (for $1 \le i < j \le d$) if its dimensions form a subset of $\{i,j\}$.

\subsection{Solution to Problem~\ref{prob:rect} in 2D}
\noindent
\begin{lemma}\label{lem:2d}
  Problem~\ref{prob:rect} for $d=2$ can be solved in $\Oh(N \log N)$ time and $\Oh(N)$ space.
\end{lemma}
\begin{proof}
  The maximal intersection point of two rectangles in 2D is either the top right corner of one of them
  or an intersection point of a vertical edge of one rectangle with a horizontal edge of the other rectangle.

  To handle the first case, we generate all top right corners of rectangles from $\R_q$ and find the one
  with the maximal $\ell_1$ norm that is contained in a rectangle from $\R_{3-q}$, for $q=1,2$.
  This requires a classical line sweep algorithm.
  Say that the sweep goes from left to right.
  The events in the sweep are the points constructed from $\R_q$ and the beginnings and endings of rectangles from $\R_{3-q}$.
  The events can be sorted in $\Oh(N)$ time thanks to the fact that the instance is normalized.
  The vertical intervals of all the rectangles that are currently intersected by the broom are
  stored in a range tree~\cite{DBLP:journals/ipl/Bentley79}.
  When a point is encountered, in $\Oh(\log N)$ time we can check if it is contained in one of the intervals
  using the range tree.
  This gives $\Oh(N \log N)$ time.

  The second case is handled using a similar line sweep.
  This time the broom stores the set of horizontal segments from $\R_q$ that it currently intersects, in a range tree ordered by their horizontal component.
  When a vertical segment resulting from $\R_{3-q}$ is encountered, the range tree can be queried in $\Oh(\log N)$ for the topmost
  horizontal segment that intersects it.
  This completes the $\Oh(N \log N)$-time algorithm.
\end{proof}

\subsection{Solution to Problem~\ref{prob:rect} in 3D}
\noindent
\begin{lemma}\label{lem:3d}
  Problem~\ref{prob:rect} for $d=3$ can be solved in $\Oh(N \log^2 N)$ time and $\Oh(N \log N)$ space.
\end{lemma}
\begin{proof}
  The maximum intersection point comes from two rectangles with the same dimensions or with different dimensions.
  The former reduces to the 2D version of the problem.
  Indeed, we iterate over all subsets of two dimensions; let us assume that the common dimensions of the rectangles are $1$ and $2$.
  We group such rectangles from $\R_1$ and $\R_2$ according to $x_3$.
  For each group, we find the maximum intersection of two rectangles using Lemma~\ref{lem:2d}.
  (The rectangles for each group can be normalized in $\Oh(N)$ time, for all the groups together.)
  This takes $\Oh(N \log N)$ time and $\Oh(N)$ space.
  
  Now let us consider the pairs of rectangles from $\R_1$ and $\R_2$ that have just one dimension in common.
  We consider all permutations of the set of components $\{x_1,x_2,x_3\}$.
  For a given permutation, we want to check for the maximal intersection point of a $12$-rectangle from $\R_q$ 
  and a $23$-rectangle from $\R_{3-q}$, for $q=1,2$.
  Let $R_1 \in \R_q$ and $R_2 \in \R_{3-q}$ be two such rectangles.
  The maximal intersection point of $R_1$ and $R_2$, if exists, has the first component equal to the first component of $R_2$,
  the third component equal to the third component of $R_1$, and the second component equal to the minimum
  of the maximal second components of $R_1$ and $R_2$.
  Without the loss of generality, we will assume that the second component is equal to the maximal second component of the rectangle from $\R_{3-q}$.
  (The opposite case will be covered when considering the components in order $x_3,x_2,x_1$.)

  The algorithm uses a plane sweep along the $x_3$ axis.
  The broom stores the rectangles from $\R_{3-q}$ that intersect with the plane.
  For such a rectangle $\{a\} \times [b_1,b_2] \times [c_1,c_2]$, we store a 2D point $(a,b_2)$
  with the weight $a+b_2$.
  When the broom encounters a rectangle $[a'_1,a'_2] \times [b'_1,b'_2] \times \{c'\}$ from $\R_q$,
  it suffices to find the point stored in the broom in the range $[a'_1,a'_2] \times [b'_1,b'_2]$
  with the maximal weight and add $c'$ to this weight.
  
  The broom can be implemented as a 2-dimensional range tree; the structure of this tree can be static,
  since the set of all potential points that are to be stored in it is known in advance.
  A single operation on the range tree costs $\Oh(\log^2 N)$ time and the total space is $\Oh(N \log N)$.
  This yields the complexity of the algorithm.
\end{proof}

\subsection{Solution to Problem~\ref{prob:rect} in 4D}
\noindent
This time we start with an auxiliary data structure.
\begin{lemma}[Interval Stabbing-Max~\cite{DBLP:journals/siamcomp/AgarwalAKMTY12}]\label{lem:intervals}
A collection of $p$ weighted intervals on a line can be stored in $\Oh(p)$ space subject
to the following operations:
  \begin{itemize}
    \item inserting a weighted interval,
    \item deleting a weighted interval, and
    \item finding an interval with the maximum weight that contains a given point,
  \end{itemize}
  each implemented in amortized $\Oh(\log p)$ time.
\end{lemma}

\begin{lemma}\label{lem:4d}
  Problem~\ref{prob:rect} for $d=4$  can be solved in $\Oh(N \log^3 N)$ time and $\Oh(N \log^2 N)$ space.
\end{lemma}
\begin{proof}
  First, assume that the sought pair of rectangles has at least one dimension in common.
  We consider all 3-element subsets of the set of dimensions $\{x_1,\ldots,x_4\}$.
  For each subset $X$, we consider the rectangles from $\R_1$ and $\R_2$ that have both their dimensions in $X$.
  We group the rectangles by equal values on all the remaining dimension and, for each group, solve an instance
  of Problem~\ref{prob:rect} for $d=3$.
  The total time complexity is $\Oh(N \log^2 N)$ and the total space complexity is $\Oh(N \log N)$, due to Lemma~\ref{lem:3d}.

  Finally, let us consider the case that the sought rectangles have no dimensions in common.
  By considering all permutations of the components, we can assume that we choose a $12$-rectangle from $\R_q$
  and a $34$-rectangle from $\R_{3-q}$, for $q \in \{1,2\}$.
  This time we will use a hyperplane sweep along the $x_4$ axis.
 
  The broom stores the rectangles from $\R_{3-q}$ that intersect with the hyperplane.
  For such a rectangle $\{a\} \times \{b\} \times [c_1,c_2] \times [d_1,d_2]$, we store the line segment $\{a\} \times \{b\} \times [c_1,c_2]$
  with the weight $a+b$.
  When the broom encounters a rectangle $[a'_1,a'_2] \times [b'_1,b'_2] \times \{c'\} \times \{d'\}$ from $\R_q$,
  it suffices to find the segment stored in the broom intersecting the rectangle $[a'_1,a'_2] \times [b'_1,b'_2] \times \{c'\}$
  with the maximal weight and add $c'+d'$ to its weight.

  The broom is implemented as a 3-dimensional range tree.
  Again, the structure of the range tree can be static if computed in advance.
  The first two dimensions correspond to the $x_1$ and $x_2$ coordinates of line segments.
  The third dimension is a data structure that stores 1D weighted intervals
  that supports the operations of: inserting a weighted interval, deleting
  a weighted interval, and finding an interval with the maximum weight that contains a given value.
  For this, the data structure of Lemma~\ref{lem:intervals} can be used.

  Inserting a line segment $\{a\} \times \{b\} \times [c_1,c_2]$ with weight $a+b$ to the data structure is straightforward.
  The query for a rectangle $[a'_1,a'_2] \times [b'_1,b'_2] \times \{c'\}$ reduces to $\Oh(\log^2 N)$ queries for the point $c'$
  in the 1D data structures.

  The total size of the data structure is $\Oh(N \log^2 N)$ and each event in the sweep is processed in $\Oh(\log^3 N)$ time.
  The conclusion follows.
\end{proof}

\subsection{Solution to Problem~\ref{prob:rect} in $d$ Dimensions}
\noindent
\begin{lemma}\label{lem:kd}
  Problem~\ref{prob:rect} can be solved in $\Oh(N d^2 \log^3 N)$ time and $\Oh(N(d^2+\log^2 N))$ space.
\end{lemma}
\begin{proof}
	The strategy behind our algorithm is to produce several instances of $d'$-dimensional Problem~\ref{prob:rect}
	with $d'\le 4$.
	Each instance is going to be created for a particular subset $D\sub \{1,\ldots,d\}$
	and its goal is to process pairs of rectangles $R_1\in\R_1$ and $R_2\in \R_2$ whose union of dimensions is precisely $D$.
	Thus, for each rectangle $R_1\in \R_1$ with dimensions $D_1$, we consider all sets $D_2$ with $|D_2|\le 2$ and insert $R_1$ to the instance corresponding to $D=D_1\cup D_2$.
	Rectangles $R_2\in \R_2$ are processed symmetrically.
	
	Next, we further subdivide each instance by grouping the rectangles according to the values at coordinates
	in the complement of $D$ so that the space dimension can be reduced from $d$ to $|D|$.
	For this, we would like to efficiently compare two coordinate vectors ignoring $|D|=\Oh(1)$ positions.
	A simple solution is to build a data structure for constant-time Longest Common Extension (LCE) queries (see \cite{AlgorithmsOnStrings}) for the concatenation of the coordinate vectors
	of all the rectangle corners. Constructing it takes linear time and space (with respect to the input size, which is $\Theta(Nd)$).
	
	Each rectangle $R$ is inserted to $\Oh(d^2)$ instances of dimension at most 4, so the overall running time to solve these instances is $\Oh(Nd^2 \log^3 N)$.
	The space complexity is $\Oh(Nd^2)$ for storing the input of the instances and $\Oh(N\log^2 N)$ for solving them one by one.
\end{proof}

\begin{theorem}
  The RLE-LCAF problem can be solved:
  \begin{itemize}
    \item in $\Oh(m^2 \sqrt{\log \log m})$ time in expectation
  or $\Oh(m^2 \log \log m)$ time deterministically and $\Oh(m)$ space for $\sigma=2$,
    \item in $\Oh(m^2 \log^2 m)$ time and $\Oh(m \log m)$ space for $\sigma=3$,
    \item in $\Oh(m^2\sigma^2 \log^3 m)$ time and $\Oh(m (\sigma^2+\log^2 m))$ space for arbitrary $\sigma$.
  \end{itemize}
\end{theorem}
\begin{proof}
  For $\sigma=2$ we use Proposition~\ref{prp:rlcaf2}.
  For $\sigma=3$ and $\sigma>3$ we use the reduction of Lemma~\ref{lem:red2} to $\Oh(m)$ instances of Problem~\ref{prob:rect} and the algorithm of
  Lemma~\ref{lem:3d} or 
  Lemma~\ref{lem:kd}, respectively.
\end{proof}

\section{Algorithm for RLE-LCAF over Large Alphabet}\label{sec:3}
It is a known and simple fact that we can assume the alphabet is $\Sigma=\{1,\ldots,m\}$; otherwise all the letters could be renumbered in $\Oh(m \log m)$ time.

We use the interpretation of the RLE-LCAF problem as a problem of intersecting rectangles from the sets $\Rect_S$ and $\Rect_T$ (Lemma~\ref{lem:red}). Recall that they consist of rectangles of the form $\rect_X(i,j)$ for $X=S,T$, respectively. 

Recall that $\I(R)$ denotes the interval of $\ell_1$-norms of points in $R$. We say that rectangles $R_1$ and $R_2$ are \emph{compatible} if $\I(R_1) \cap \I(R_2) \ne \emptyset$. We further say that $R_2$ is \emph{max-compatible} with $R_1$ if $\max(\I(R_1)) \in \I(R_2)$. The following basic observation characterizes these notions.

\begin{observation}\label{obs:compat}
  Let $R_1$ and $R_2$ be rectangles.
  \begin{enumerate}[(a)]
    \item If $R_1 \cap R_2 \ne \emptyset$, then $R_1$ and $R_2$ are compatible.
    \item $R_1$ and $R_2$ are compatible if and only if $R_1$ is max-compatible with $R_2$ or $R_2$ is max-compatible with $R_1$.
  \end{enumerate}
\end{observation}

For a rectangle $R$, by $\P(R)$ we denote the components of any point in $R$,
with `\placeholder' at the coordinates that correspond to the dimensions of $R$. Thus $\P(R)$ represents the coordinates of all points in $R$.
In similarity with the original problem, we call $\P(R)$ the \emph{Parikh vector} of the rectangle $R$. 
We call two rectangles $R_1$ and $R_2$ \emph{consistent} if $\P(R_1)$ and $\P(R_2)$ are equal for each coordinate where \placeholder does not appear in either vector.
By $\Delta(R_1,R_2)$ we denote the intervals of $R_1$ and $R_2$ that correspond to the coordinates where at least one of the Parikh vectors $\P(R_1)$, $\P(R_2)$ contains a \placeholder. 

\begin{example} The rectangles
\begin{align*}
R_1&=\{5\} \times [1,3] \times \{4\} \times [1,6] \times \{3\}\\
R_2&=[4,5] \times \{2\} \times \{4\} \times [2,5] \times \{3\}
\end{align*}
are consistent. We have $\P(R_1)=(5,\placeholder,\underline{4},\placeholder,\underline{3})$ and $\P(R_2)=(\placeholder,2,\underline{4},\placeholder,\underline{3})$. Here the set $\Delta(R_1,R_2)$ stores the intervals $\{5\}$, $[1,3]$, and $[1,6]$ from $R_1$ and $[4,5]$, $\{2\}$, and $[2,5]$ from $R_2$ that correspond to the coordinates 1, 2, and 4.
\end{example}

Let us make the following observation. Point (a) of the observation is straightforward. Point (b) boils down to simple arithmetics (see also \cite{DBLP:conf/spire/Grabowski17}).

\begin{observation}\label{obs:result}
  Let $R_1$ and $R_2$ be rectangles.
  \begin{enumerate}[(a)]
    \item If $R_1 \cap R_2 \ne \emptyset$, then $R_1$ and $R_2$ are consistent.
    \item If $R_1$ and $R_2$ are consistent, then knowing $\I(R_1)$, $\I(R_2)$, and $\Delta(R_1,R_2)$, one can compute the maximum $\ell_1$-norm of a point in $R_1 \cap R_2$, if it exists.
  \end{enumerate}
\end{observation}

The following lemma forms the basis of our algorithm. It follows from the properties of the sets $\Rect_V(l)$ for $V \in \{S,T\}$ (Lemma~\ref{lem:rect-l}\eqref{aaa}).

\begin{lemma}\label{lem:workhorse}
  Let $R_1 \in \Rect_S$ and assume that $\P(R_1)$ is known.
  \begin{enumerate}[(a)]
    \item The set of all $R_2 \in \Rect_T$ that are max-compatible with $R_1$ has size at most $2|T|$.
    \item The values $\I(R_2)$ and $\Delta(R_1,R_2)$ for all $R_2 \in \Rect_T$ that are max-compatible and consistent with $R_1$ can be computed in $\Oh(m)$ time and space.
  \end{enumerate}
\end{lemma}
\begin{proof}
  We start by computing $\I(R_1)$. Let $k=\max \I(R_1)$. Let us recall that $R_2$ is max-compatible with $R_1$ if and only if $k \in \I(R_2)$. The set of all such rectangles is precisely $\Rect_T(k)$. By Lemma~\ref{lem:rect-l}\eqref{aaa}, this set has size at most $2|T|$. This concludes point~(a).

  Lemma~\ref{lem:rect-l}\eqref{aaa} asserts that the set $\Rect_T(k)$ can be computed in $\Oh(m)$ time. To implement point (b), we will use a sliding window approach. We iterate over all indexes $i=1,\ldots,m$ and for each of them we consider $j=j(i,k),\ldots,j(i+1,k)$ as in Observation~\ref{obs:rect-l}. When any index $i$ or $j$ is incremented, $\P(R_2)$ for $R_2=\rect_T(i,j)$ can be updated in $\Oh(1)$ time, starting from an initial Parikh vector with all zeros. This allows us to store and update $\I(R_2)$ and $\Delta(R_1,R_2)$. We also store a counter $q$ of positions $a \in \{1,\ldots,\sigma\}$ such that $\P(R_1)$ and $\P(R_2)$ differ at position $a$ and $a$ is not a dimension of any of the two rectangles. This counter can be updated in $\Oh(1)$ time whenever $i$ and $j$ is incremented by inspecting the positions in $\P(R_2)$ that have changed. Then $R_2$ is consistent with $R_1$ if and only if $q=0$. 
  
  The computations made upon an incrementation of $i$ or $j$ take $\Oh(1)$ time; the time complexity follows. The only additional space used in the algorithm is the Parikh vector which takes $\Oh(\sigma)=\Oh(m)$ space.
\end{proof}

We arrive at the main result of this section.

\begin{theorem}
  The RLE-LCAF problem can be solved in $\Oh(m^3)$ time and $\Oh(m)$ space.
\end{theorem}
\begin{proof}
  We apply the reduction to Problem~\ref{prob:rect} of Lemma~\ref{lem:red}.
  
  For each $i=1,\ldots,|S|$ we generate $\P(\rect_S(i,j))$ for consecutive $j=i,\ldots,|S|$. The next Parikh vector overwrites the previous one. Thus these Parikh vectors can be computed one by one in $\Oh(m)$ total time and space. Whenever a new Parikh vector is generated, we use Lemma~\ref{lem:workhorse} to compute the values $\I(R_2)$ and $\Delta(R_1,R_2)$ for all $R_2 \in \Rect_T$ that are max-compatible and consistent with $R_1$. Finally, we use Observation~\ref{obs:result}(b) to compute the maximum $\ell_1$ norm of a point in $R_1 \cap R_2$, for any of the rectangles $R_2$.
  
  Next we repeat the whole procedure with $S$ and $T$ interchanged.
  
  Correctness of the algorithm follows from Observation~\ref{obs:compat}: if $R_1\in \Rect_S$ and $R_2 \in \Rect_T$ intersect, then they are compatible, which means that $R_2$ is max-compatible with $R_1$ or $R_1$ is max-compatible with $R_2$.
  
  There are $\Oh(m^2)$ rectangles $R_1 \in \Rect_S$ and for each of them the computations take $\Oh(m)$ time. Hence, the algorithm works in $\Oh(m^3)$ time. The space complexity is linear.
\end{proof}

\section{Conclusions and Open Problems}
\noindent
We have presented efficient algorithms for the LCAF and RLE-LCAF problems:
\begin{itemize}
  \item $\Oh(n^2 / \log^{1+1/\sigma} n)$-time and $\Oh(n)$-space algorithm for LCAF with $\sigma=\Oh(1)$;
  \item $\Oh(m^2\sigma^2 \log^3 m)$ time and $\Oh(m (\sigma^2+\log^2 m))$ space algorithm for RLE-LCAF with arbitrary $\sigma$;
  \item $\Oh(m^3)$-time and $\Oh(m)$-space algorithm for RLE-LCAF with arbitrary $\sigma$.
\end{itemize}
For LCAF over a constant-sized alphabet, we have obtained an over-logarithmic speedup comparing to a naive $\Oh(n^2)$-time solution. Let us recall that over the binary alphabet, LCAF can be solved much more efficiently, in $\Oh(n^{1.859})$ time~\cite{DBLP:journals/ijfcs/AlatabbiILR16,DBLP:conf/stoc/ChanL15}. An open question is to design an $\Oh(n^{2-\varepsilon})$-time algorithm ($\varepsilon>0$) for LCAF for alphabet of any constant size, e.g., for $\sigma=3$.

For the RLE-LCAF problem, the most interesting question is for the existence of an $\Oh(m^{3-\varepsilon})$-time and $\Oh(m)$ space algorithm for an arbitrary $\sigma$.

\paragraph{\bf Acknowledgements}
Jakub Radoszewski was supported by the “Algorithms for text processing with errors and uncertainties” project carried out within the HOMING program of the Foundation for Polish Science co-financed by the European Union under the European Regional Development Fund.

\bibliographystyle{abbrv}
\bibliography{abelian_lcf}
\end{document}